\documentclass[12pt]{article}
\usepackage{fullpage}
\usepackage{times,amsmath,epsfig,amssymb,amsthm}

\newtheorem{remark}{Remark}
\newtheorem{problem}{Problem}

\newtheorem{proposition}{Proposition}

\newtheorem{claim}{Claim}
\newtheorem{definition}{Definition}

\newtheorem{lemma}{Lemma}
\newtheorem{theorem}{Theorem}
\newcommand{\sdeg}{{\sf syntdeg}}

\newcommand{\rej}{\#{\sf rej}}
\newcommand{\gap}{{\sf gap}}
\newcommand{\acc}{\#{\sf acc}}

\newcommand{\W}{{\sf W}}
\newcommand{\fpt}{{\sf FPT}}
\newcommand{\rfpt} {\W[{\sf P}]{\mbox{-}{\sf RFPT}}}

\newcommand{\pperm}{{\sf p}\mbox{-}{\sf perm}}
\newcommand{\pdet}{{\sf p}\mbox{-}{\sf det}}
\newcommand{\pacit}{{\sf p}\mbox{-}{\sf acit}}
\newcommand{\w}[1]{{\sf W}[#1]}
\newcommand{\Gap}{{\sf Gap}\text{-}}
\newcommand{\diff}{{\sf Diff}\text{-}}
\newcommand{\Wp}{{\sf W[P]}}
\newcommand{\ShP}{\#{\sf P}}
\newcommand{\ppp}{{\sf W[P]}\mbox{-}{\sf PFPT}}

\newcommand{\size}{{\sf size}}
\begin{document}
\title{Parameterized Analogues of Probabilistic Computation\thanks{   Department of Computer Science and Engineering
  Indian Institute of Technology Madras, Chennai, India. 
  \tt{\{ankitch,bvrr\}@cse.iitm.ac.in}}}
\author{Ankit Chauhan
        \and 
        B. V. Raghavendra Rao}  

\maketitle

\begin{abstract}
We study structural aspects of randomized  parameterized computation. We  introduce a new class $\ppp$ as a natural parameterized analogue of ${\sf PP}$. Our definition uses the  machine based characterization of the parameterized complexity class $\Wp$ obtained by Chen et.al [TCS 2005].   We translate most of the structural properties and characterizations of the class ${\sf PP}$ to the new class  $\ppp$.  

We study a parameterization of the polynomial identity testing problem based on the degree of the  polynomial computed by the arithmetic circuit.   We obtain a parameterized analogue of the well known Schwartz-Zippel lemma [Schwartz, JACM 80 and Zippel, EUROSAM 79].  
 
Additionally, we introduce a parameterized variant of  permanent, and  prove its $\#W[1]$ completeness.
\end{abstract}
\section{Introduction}
{\em Parameterized Complexity Theory} provides a formal framework for finer complexity analysis of problems by allowing a parameter along with the input. It was pioneered by Downey and Fellows~\cite{DF,DF92} two decades ago. Since then, it has revolutionized algorithmic research~\cite{Nie06}, and  led to the development of several important algorithmic techniques.  

{\em Fixed Parameter Tractability} (FPT) forms the central notion of tractability in Parameterized Complexity Theory. Here, any problem  that is decidable in deterministic time $f(k){\sf poly}(n)$ is deemed to be tractable, where $k$ is the parameter  and $f$ any computable function. Several ${\sf NP}$ hard problems including the vertex cover problem are known to be tractable under this notion~\cite{FG}. 
 
 The $\W$-hierarchy serves as the basis for all intractable problems in the parameterized world. $\W[1]$, the smallest member of $\W$-hierarchy, consists of  problems that are FPT equivalent to the $p$-clique problem~\cite{FG}. The limit of  $\W$ hierarchy, $\Wp$ encapsulates all problems solvable in  non-deterministic $f(k){\sf poly}(n)$ time using at most $g(k)\log n$  non-deterministic bits~\cite{CFG05,FG}, where $f$ and $g$ are arbitrary computable functions.

There have been significant efforts towards  understanding the structure of parameterized complexity classes in the last two decades.  Specifically exact characterizations of the $\W$ hierarchy and other related hierarchies are known~\cite{DFR98}. (See also \cite{FG,DF}.)
   
   Apart from non-deterministic computation, probabilistic computation serves as  one of the crucial  building blocks of Complexity Theory. Probabilistic complexity classes have been well studied in the literature and has been an active area of research for more than three decades.  There are a  significant number of parameterized algorithms  that use randomization~\cite[Chapter 8]{DF}.   Hence, development of randomized complexity classes in the parameterized framework is necessary to understand the use of randomization in the parameterized setting.
          
M\"uller~\cite{Mul08,Mul08b}  was the first to  do a systematic development and  study of parameterized randomization. He defined bounded error probabilistic parameterized classes such as ${\Wp}\mbox{-}{\sf BPFPT}$ and ${\W[1]}\mbox{-}{\sf BPFPT}$.  Further, he obtained  amplification results and  conditions for derandomization of these classes.    Further, M\"uller~\cite{Mul08}  studied several parameterizations of the well known polynomial identity testing problem (ACIT) and obtained several hardness results as well has  upper bounds in terms of the newly defined randomized classes. 

 We continue the line of research initiated by M\"uller~\cite{Mul08} and study a  parameterized variant of probabilistic computation with unbounded error and  establish a relationship with the corresponding parameterized counting class. 
 
  It should be noted that almost all of the randomized FPT algorithms use randomness of the same magnitude  as their running times. However, such an algorithm  cannot be visualized as a non-deterministic algorithm with $f(k)\log n$ random bits, where $f(k)$ is an arbitrary computable function. This is in stark contrast to the classical setting, where every randomized algorithm with bounded error probability can also be seen as a non-deterministic algorithm with the same time bound. So it is desirable to have randomized FPT algorithms that use at most $O(f(k)\log n)$ random bits instead of $f(k) {\sf poly}(n)$ random bits. As a first step towards this we obtain such an algorithm for a suitable parameterization of ACIT.          
 
Finally, following the recent developments in the parameterized complexity theory of counting problems~\cite{BC12,Cur13,CM14}, we develop a parameterized  variants of the problems of computing permanent and determinant of a matrix.

\paragraph*{Our results}
We make an attempt at understanding the relations between counting and probabilistic classes. We focus on a probabilistic analogue of the  class $\Wp$. Using the notion of $k$-restricted Turing machines~\cite{CFG05}, we introduce $\ppp$ as a parameterized variant of the probabilistic polynomial time ({\sf PP}). As in the classical complexity setting, we  establish a close connection between $\ppp$  and  the counting class~$\#\Wp$ (Theorem~\ref{thm:oracle-wp}). Further, we  show that $\ppp$ is closed under complementation and symmetric differences. (Theorem~\ref{thm:pp-comp} and Lemma~\ref{lem:pp-symdiff}.)

We consider the polynomial identity testing problems ({\sf ACIT}) with the {\em syntactic degree} (See Section~2 for a definition) as a parameter. Using the construction of hitting set generators by Shpilka and Volkovich~\cite{SV09}, we obtain what can be called as a  parameterized analogue of the celebrated Schwartz-Zippel Lemma~\cite{Sch80,Zip79}. (Theorem~\ref{thm:pacit}.)

Finally, we introduce a parameterized variant of the permanent function $\pperm$ and prove that it characterizes the class $\#\W[1]$. (Theorem~\ref{thm:pperm}.) Analogously, a variant of the determinant  function ($\pdet$) and show that it is Fixed Parameter Tractable (Theorem~\ref{thm:pdet}).
 
\section{Preliminaries}
 We include some of the definitions from Parameterized Complexity theory and  Complexity theory  here.
For Parameterized Complexity, the notations in~\cite{DF,FG} are followed. Definitions of complexity classes   can be found in e.g.,~\cite{DK,AB}. 
  

 A {\em parameterized}  language   is a set $P \subseteq \Sigma^* \times \mathbb{N} $, where $\Sigma$ is a finite alphabet. If $( x, k ) \in \Sigma^* \times \mathbb{N}$ is an input  instance of a parameterized language, then  $x$ is referred to  as the {\em input} and  $k$ as the {\em  parameter}.
 
  A {\em  parameterized counting} problem is a  pair $(f,k)$, where  $f: \Sigma^* \to \mathbb{N} $ is a counting function and $k$ is the parameter and $\Sigma$ is  a finite alphabet. For notational convenience, we will denote a parameterized counting problem as a function $f:\Sigma^*\times \mathbb{N}\to\mathbb{N}$, where the second argument to $f$ is considered as the parameter.

A parameterized language $P \subseteq \Sigma^* \times \mathbb{N} $ is said to be {\em  fixed-parameter tractable} if there is an algorithm that given a pair $(x, k) \in   \Sigma^* \times \mathbb{N}$ , decides if $(x, k) \in P$ in at most $O(f(k)|x|^c)$ steps, where  $f:\mathbb{N}\rightarrow \mathbb{N}$ is a computable function and  $c\in \mathbb{N}$ is a constant.
\begin{definition} \fpt\ denotes        the  class of all parameterized languages that are fixed parameterized tractable. 
\end{definition}

A parameterized language $L$ is said to be in {\sf RFPT} (Randomized \fpt) if there is a $f(k){{\sf poly}(n)}$ time bounded  randomized machine  accepting $L$ with bounded one-sided error probability. See~\cite{FG,DF} for more details.

\begin{definition}
A {\em $k$-restricted  machine} is a  non-deterministic   $g(k){\sf poly}(n)$ time bounded  Random Access Machine (RAM) that uses at most $f(k)$ non-deterministic words, where $f$ and $g$ are arbitrary computable functions. Here we assume that the word size is $O(\log n)$, where $n$ is the length of the input.

A {\em $k$-restricted Turing machine} is a non-deterministic   $g(k){\sf poly}(n)$ time  Turing machine that makes at most $f(k)\log n$ non-deterministic moves, where $f$ and $g$ are arbitrary computable functions. 
\end{definition}
\begin{definition}
A {\em tail} non-deterministic  machine is a $k$-restricted machine in which all non-deterministic steps are among last $f(k)$ steps.  
\end{definition}
$\W$[P] is the class of all parameterized problems ($Q,k$) that can be decided by a $k$-restricted non-deterministic  machine (for more details see chapter 3 in \cite{FG}).
$\W$[1] is the class of all parameterized problems ($Q,k$) that can be decided by {\em tail} non-deterministic  machine (for more details see \cite{CFG05}).

 
 For a non-deterministic machine $M$, let $\acc_M(x,k)$ and $\rej_M(x,k)$ respectively denote the number of accepting and rejecting paths of $M$ on input $(x,k)$. Define 
 $\gap_M(x,k){=}\acc_M(x,k)-\rej_M(x,k).$

\begin{definition}\cite{FG}
A parameterized counting function $(f,k)$ over the alphabet $\Sigma$ is in $\#\W$[P] if there is a $k$-restricted non-deterministic  machine
$M$ such that  $f(x,k)=\acc_M(x,k)$. 
\end{definition}
\begin{definition}
A probabilistic $k$-restricted   machine is a probabilistic  $g(k){\sf poly}(n)$ time bounded  RAM that make at most $f(k)$ probabilistic moves, where $f$ and $g$ are some computable functions. Here we assume that one probabilistic move involves choosing a random word of $O(\log n)$ bits.   
\end{definition}
A language $L$ is said to be in $\rfpt$~\cite{Mul08} if there is a $k$-restricted probabilistic  machine such that $(x,k)\in L\implies Pr[M \mbox{ accepts } (x,k)]\ge 2/3$ ; \mbox{ and } $x\notin L \implies Pr[M \mbox{ rejects } (x,k)]= 0$.

An {\em arithmetic circuit} $C$ is a directed acyclic graph with labelling on the vertices as follows. Nodes  of   in-degree zero are  called {\em input} gates and  are labelled from $\{-1,0,1\}\cup\{x_1,\ldots, x_n\}$ where $x_1,\ldots, x_n$ are the input variables. The remaining  gates  are  labelled  $\times$ or $+$. An arithmetic circuit has exactly one gate of zero out-degree called the {\em output} gate. Every gate $v$ in an arithmetic circuit can naturally be associated with a polynomial $p_v\in \mathbb{Z}[x_1,\ldots,x_n]$, where the polynomials associated at input nodes are either constants or variables. If $v=v_1+v_2$ then $p_v= p_{v_1}+p_{v_2}$ and if $v=v_1\times v_2$ then $p_v=p_{v_1}\times p_{v_2}$.  The polynomial computed by the circuit $C$ is the polynomial associated with its only output gate and is denoted by $p_C$.
The size of an arithmetic circuit is the number of gates in it and is denoted by $\size(C)$.       

We associate a number called the {\em syntactic degree}  ({\sf syntdeg})\footnote{Syntactic degree is also known as the formal degree~\cite{KSS14} and is a standard parameter for arithmetic circuits.} with every gate of an arithmetic circuit $C$. For a leaf node $v$, $\sdeg(v)=1$. If $v=v_1+v_2$ then $\sdeg(v)=\max\{\sdeg(v_1),\sdeg(v_2)\}$ and if $v=v_1\times v_2$ then
$\deg(v)=\sdeg(v_1)+\sdeg(v_2)$.  It should be noted that the degree of the polynomial computed by a circuit is bounded by its syntactic degree. 
\begin{remark}
 the parameter $d_{\times}$ introduced in~\cite{Mul08} is closely related to $\sdeg$, in fact $\sdeg\le 2^{d_{\times}}\le 2^{\sdeg}$. 
\end{remark}
In~\cite{Alo99}, Alon  obtained a characterization for multivariate  polynomials
that are not identically zero known as the Combinatorial Nullstellensatz: 
\begin{proposition}[Combinatorial Nullstellensatz, \cite{Alo99}]
\label{pro:combnull}
Let $P\in \mathbb{K}[x_1, \ldots, x_n]$ be a polynomial where for
every $i\in [n]$, the degree of $x_i$ is bounded by $t$. Let
$S\subseteq\mathbb{K}$ be a finite set of size at least $t+1$, and $A=S^n$. Then
$P\equiv 0 \iff P(a) = 0, ~\forall a\in A$.
\end{proposition}

\section{Probabilistic Computation}
In this section, we develop a parameterized analogue of the  classical complexity class ${\sf PP}$.  Our definition of $\ppp$ is based on   $k$-restricted probabilistic Turing machines.  

Throughout this section unless otherwise stated, $f(k)$ denotes an arbitrary computable function, and $P(n,k)=f(k)\log n$.  For an input $x$, we denote $n=|x|$.

\begin{definition}
Let $L$ be a parameterized language. $L$ is said to be in the class $ \ppp$ if there is a $k$-restricted  probabilistic Turing machine $M$ such that for any $(x,k)\in \Sigma^*\times\mathbb{N}$ we have,
\begin{eqnarray*}
(x,k) \in L \Rightarrow \Pr[ M\mbox{ accepts }(x,k) ]> \frac{1}{2} 
\\
(x,k) \notin L \Rightarrow \Pr[  M \mbox{ accepts }(x,k) ]\leq \frac{1}{2}
\end{eqnarray*}
where the probabilities are over the random choices made by $M$.
\end{definition}
 Without loss of generality,  we assume  $\Sigma=\{0,1\}$.

In the classical setting, ${\sf PP}$  is known to   have several characterizations  based on, 1) difference between two $\#{\sf P}$ functions~\cite{For97}, 2) difference between the number of accepting  and rejecting paths of a polynomial time bounded non-deterministic  Turing machine~\cite{For97}, 3) logics based on majority quantifiers~\cite{Kon09} and 4)  large fan-in circuits   with threshold gates~\cite{AW90}. We observe that all of the characterizations except (3) hold for $\ppp$. However, it is not clear if the majority quantifier logical characterization of ${\sf PP}$~\cite{Kon09}  translates to the parameterized setting.
  
\begin{definition}[\diff\fpt, \Gap\fpt]
A  parameterized function $f: \Sigma^*\times \mathbb{N}\to {\mathbb{Z}} $ is said to be in $\diff\fpt$ if there are two  functions $g,h\in \#\Wp$ such that $f(x,k)= g(x,k)-h(x,k)$.

$f$ is said to be in $\Gap\fpt$ if there is a $k$-restricted TM $M$ such that $f(x,k)=\acc_M(x,k)-\rej_M(x,k)$, $\forall (x,k)\in \Sigma^*\times \mathbb{N}$.
\end{definition}
 Firstly, we observe that the two classes $\Gap\fpt$ and $\diff\fpt$ coincide.
\begin{lemma}
\label{lem:diff-gap}
$\Gap\fpt=\diff\fpt$
\end{lemma}

 \begin{proof}
To show $\Gap\fpt\subseteq \diff\fpt$: Let $f\in\Gap\fpt$, then there is a $k$-restricted $M$ with $f(x,k)= \acc_M(x,k)-\rej_M(x,k)$.  Let $M'$ be a new machine that  simulates $M$ on input $(x,k)$ and accepts if and only if $M$ rejects $(x,k)$.
Then we have $f(x,k)=\acc_M(x,k)-\acc_{M'}(x,k)$.
For the converse inclusion, let $f\in \diff\fpt$, and $M_1, M_2$ be such that $f(x,k)= \acc_{M_1}(x,k)-\acc_{M_2}(x,k)$.
 Let $M$ be  a new machine: on input $(x,k)$,  $M$ runs $M_1$ on $(x,k)$, and accepts if $M_1$ does so. If $M_1$ rejects then $M$ simulates $M_2$ on $(x,k)$ and rejects if $M_2$ accepts. If $M_2$ rejects, then $M$ guesses a non-deterministic bit $b$, accepts if $b=1$ and rejects otherwise. Then $\acc_M(x,k)-\rej_M(x,k)=\acc_{M_1}(x,k)-\acc_{M_2}(x,k)=f(x,k)$. 
\end{proof}
\begin{lemma}
\label{lem:gap-closure}
$\Gap\fpt$ is closed under taking {\em $p$-bounded}  summations and products, i.e., if $g_1,\ldots, g_{t(k)}\in \Gap\fpt$,
then so are $g_1+g_2 \cdots + g_{t(k)}$ and $g_1\times g_2 \times \cdots \times g_{t(k)}$, where $t$ is any computable function.
\end{lemma}
\begin{proof}
The arguments here are straightforward  adaptations of  proofs from classical complexity. We include it here for completeness. 
For summation, we can construct a new machine $M$ that first guesses  $i\in [1,t(k)]$ and and runs the $k$-restricted machine for $g_i$ on $(x,k)$. 

For product,  we will show for the case  when $t(k)=2$. Let $f_1,f_2\in \Gap\fpt$. Let  $M_1$ and $M_2$ as the $k$-restricted machines   such that
$f_i(x,k)= \acc_{M_i}(x,k)- \rej_{M_i}(x,k)$, $1\le i\le 2$. Let $\overline{M_i}$ be the machine that flips the answers of $M_i$. Let $M$ be $k$-restricted  machine defined as follows: On input $(x,k)$ first simulate $M_1$ on $(x,k)$. If $M_1$ accepts then run $M_2$ on $(x,k)$ and accept if and only if $M_2$ does so. If $M_1$ rejects then run  $\overline{M_2}$ on $(x,k)$ and accept if and only if $\overline{M_2}$ does so.  It can be seen that $f_1(x,k)f_2(x,k)= \acc_M(x,k)-\rej_M(x,k)$.

The above argument can be generalized to the case $t(k)\ge 2$. 
\end{proof}

  \begin{theorem}
  \label{thm:char-pp}
   Let $L$ be a parameterized language. The following are equivalent:
  \begin{enumerate}
  \item $L\in \ppp$.
  \item There is a  $k-$restricted Turing machine $M$ such that,

  $(x,k)\in L\iff \#{\sf accept}_M(x,k) -\#{\sf reject}_M(x,k)>0$ .
  \item There is a function $f\in \Gap\fpt$ such that $(x,k)\in L \iff f(x,k)>0$
\item There is a $B\in \fpt$, and $P(n,k)=f(k)\log n$ such that
$(x,k)\in L \iff |\{y\in \{0,1\}^{P(n,k)}~|~(x,y,k)\in B\}|\ge 2^{P(n,k)-1}+1$.
\end{enumerate}   
  \end{theorem}
\begin{proof}[Theorem~\ref{thm:char-pp}]
  ($1\Rightarrow 2$)  Let $L\in \ppp$. Let $M$ be a $k$-restricted probabilistic  machine for $L$. Then,
  \begin{eqnarray*}
 (x, k) \in L &\Rightarrow& Pr[\text{M accept} (x, k)] >\frac{1}{2} \Rightarrow \frac{\#{\sf accept}_M(x,k)}{\#{\sf accept}_M(x,k)+\#{\sf reject}_M(x,k)}> \frac{1}{2}\\ 
     &\Rightarrow& \#{\sf accept}_M(x,k) - \#{\sf reject}_M(x,k) > 0 
 \end{eqnarray*}  

 $(2\Rightarrow 3)$    This directly follows from the definition of $\Gap\fpt$. 
 
 $(3\Rightarrow 4)$ Let $f\in \Gap\fpt$ with $(x,k)\in L\iff f(x,k)>0$, and $M$ be a $k$-restricted machine with $f(x,k)=\gap_M(x,k). $ Let $P(n,k)$ be the number of non-deterministic bits used by $M$ on an input of length $n$ with parameter $k$.  Then  $\gap_M(x,k)>0 \implies  \acc_M(x,k) > 2^{P(n,k)}/2 =2^{P(n,k)-1}$. Let
  $$B= \{\langle x, y, k \rangle\mid \mbox{ $M$ on the  non-deterministic path defined by $y$ accepts $x$.} \} $$
Clearly, $B\in \fpt$ and 
$$ \acc_M(x,k) = |\{ y\in \{0,1\}^{P(n,k)}~|~ \langle x,y,k \rangle\in B \}|. $$

Thus $(x,k)\in L \implies|\{ y\in \{0,1\}^{P(n,k)}~|~ \langle x,y,k \rangle\in B \}|>2^{P(n,k)-1}$. 

$(4 \Rightarrow 1)$  Let $L$ as given in 4. Let $M$ be $k$-restricted machine that on input $(x,k)$ guesses a string $y\in \{0,1\}^{P(n,k)}$ and  accepts if and only if $\langle x,y,k \rangle\in B$. Then we have $x\in L \iff \acc_M(x,k)> 2^{P(n,k)-1} \iff Pr[M~\mbox{accepts } (x,k)]>1/2$.

\end{proof}
 
Similar to the case of ${\sf PP}$, we observe that an FPT machine
with oracle access to a function in $\#W[P]$ is equivalent to an FPT machine with a language in $\ppp$ as an oracle.  
\begin{theorem}
\label{thm:oracle-wp}
 $\fpt^{\#\Wp}=\fpt^{\ppp}$
\end{theorem}  
  \begin{proof}[Theorem~\ref{thm:oracle-wp}]
We show the containment in both the directions. 
We start with the easier direction, i.e., we show
$\fpt^{\ppp}\subseteq \fpt^{\#\Wp}$.

Let $L\in \fpt^{\ppp}$ and $M$ be a  deterministic oracle   Turing machine that runs in time $f(k){\sf poly}(n)$  and  $A\in \ppp$  be such that such that $L= L(M^{A})$. We need to show that $L\in FPT^{\#W[P]}$. By  Theorem~\ref{thm:char-pp}, there are two parameterized  functions $g,h:\{0,1\}^*\times k\to {\mathbb N}$ with $g,h\in \#\Wp$ such that 
\begin{eqnarray}
(x,k)\in A\iff g(x,k)- h(x,k)>0.
\label{eq:pp-gap}
\end{eqnarray}
 Let $\gamma:\{0,1\}^*\times k \to \mathbb{N} $  where
$ \gamma(0x,k)= g(x,k), \mbox{ and } \gamma(1x,k) = h(x,k),  \forall x\in \{0,1\}^*$.
On strings of length $0$ and  $1$, $\gamma$ can be defined arbitrarily.  We have $\gamma \in \#\Wp$,  since a  on input $y=ax$, with $a\in \{0,1\}$, the machine would run the machine for $g$ if $a=0$ and machine for $h$ if $a=1$.

We can simulate a query  $(y,k')$ made by the machine $M$ to $A$ by two queries to the function $\gamma$: (1) Query $(0y,k')$ and (2) $(1y,k')$, compute the difference of the values obtained and  use (\ref{eq:pp-gap}) to decide the membership of $(y,k')$ in $A$. Thus we can conclude $L\in\fpt^{\#\Wp}$.

For the reverse containment, 
  given a Turing machine $M$, let  $L_M$  be the language defined as : $ L_M= \{ ((x,k, y) \in \Sigma^*\times \mathbb{N}  \times \mathbb{N} ~|~ \acc_{M}(x,k) > y \} $
\begin{claim}
 Let $M$ be a $k$-restricted Turing machine, then $L_M\in \ppp$.
 \end{claim}  
\begin{proof}[Claim]
Let $M'$ be a Turing machine computing function $t(x,k,y)$, that on input $(x,k, y)$, produces exactly $y$ accepting paths, where $y$ is represented in binary, and $y\in[0, 2^{P(n,k)}]$.  Clearly, $M'$ is a $k$-restricted Turing machine, since
it needs to use only $P(n,k)$ many non-deterministic bits. Thus the function $t(x,k,y)=y$ is in $\#\Wp$. Let $f_M(x,k,y)=\acc_M(x,k)-y$. Then by Lemma~\ref{lem:diff-gap} $f_M$ is in $\gap\Wp$ and the claim now follows from Theorem~\ref{thm:char-pp}.
\qed\end{proof}
Let $L\in \fpt^{\#\Wp}$, then there is a deterministic oracle Turing machine $M'$ that runs $\fpt$ time, and a function $g\in \#\Wp$
such that  $L= M'^{g}$. Let $M$ be a $k$-restricted Turing machine that uses at most $f(k)\log n$ non-deterministic steps such that $g(x,k)=\acc_M(x,k)$.   We use the standard binary search technique to show that  $g(x,k)$ can be computed  using $O(kn)$ many queries to the language $L_M$.
 \begin{description}
 \item[Input] $(x,k)$, oracle access to $L_M$. {\bf Output} $g(x,k)$. 
 \item[1.]  Initialize $p= P(|x|,k)$, $y=2^p$.
\item[2.] Repeat steps 3 \& 4 until $p\ge 0$ 
 \item[3.] Query $(x,k,y)$ to the oracle;  If YES, then set $b_p=1$ and  $y=y+2^{p-1}$;  Else  
  set $b_p=0$.
  \item[4.] Set $p=p-1$
  \item[5.] Return $a={\sf binary}(b_pb_{p-1}\ldots b_0)$.
 \end{description}
 In the above ${\sf binary}(b_pb_{p-1}\ldots b_0)=\sum_{i=0}^p 2^ib_i$. Clearly, the algorithm  above runs in time $f(k){\sf poly}(n)$, and hence computing $g$ can be done in $\fpt$ with oracle access to $L_M\in \ppp$. This concludes  the inclusion in the converse direction.  
\qed \end{proof}
\begin{theorem}
\label{thm:pp-comp} 
$\ppp$ is closed under complementation.
\end{theorem}
Proof can be found in the appendix.

\begin{lemma}
\label{lem:pp-symdiff}
$\ppp$ is closed under symmetric difference. 
\end{lemma}
\begin{proof}
The proof  essentially follows the ideas in the classical setting~\cite{AW90}. 
Let $L_1,L_2 \in \ppp$. By Theorem~\ref{thm:char-pp}, there are  languages $B_1, B_2 \in \fpt$, and a function  $P(n,k)=f(k)\log n$ such that for any $x\in \{0,1\}^*$,  $k\in \mathbb{N}$ and $i\in\{1,2\}$,
\[(x,k) \in L_i \iff |\{ y_i\in \{0,1\}^{P(n,k)}|(x,k,y_i)\in B_i \}| \geq 2^{P(n,k)-1}+1 \]
Using a construction similar to the one used in the proof of Theorem~\ref{thm:pp-comp}, we get parameterized languages $B_{1}'$ and $B_{2}'$, and a function $P'(n,k)=f'(k)\log n$  with the following property for $1\le i\le 2$:
\begin{eqnarray*}
(x,k) \in L_i &\implies& |\{ y_i\in \{0,1\}^{P'(n,k)}|(x,k,y_i)\in B_{i}' \}| \ge 2^{P'(n,k)-1} + 1; \mbox{and} \\
(x,k) \notin L_i &\implies& |\{ y_i\in \{0,1\}^{P'(n,k)}|(x,k,y_i)\in B_{i}' \}| \le 2^{P'(n,k)-1} - 1.
\end{eqnarray*}
Let $a_1(x), a_2(x)\in\mathbb{Z}$ such that
$|\{y \in\{0,1\}^{P(n,k)}~|~ \langle x,y\rangle \in B_{i}'\}|=2^{P(n,k)-1}+a_i(x)$ for $1\le i\le 2$. Thus $|\{y\in \{0,1\}^{P(n,k)}~|~ \langle x,y\rangle \notin B_{i}'\}|=2^{P(n,k)-1}-a_i(x)$.  
For $x\in\Sigma^*$, let 
\begin{eqnarray*}
\ell(x,k) & &\stackrel{\triangle}{=}|S(x,k)|\\ & &= |(2^{P'(n,k)-1}+a_1)(2^{P'(n,k)-1}-a_2)+(2^{P'(n,k)-1}-a_1)(2^{P'(n,k)-1}+a_2)|\\ 
 & &=(2^{2P'(n,k)-1}-a_1a_2) ,
\end{eqnarray*}
where $S(x,k)= \{\langle y_1,y_2\rangle \mid (<x,y_1>\in B_1 \wedge \langle x,y_2\rangle \notin B_2)\vee (\langle x,y_1\rangle\notin B_1 \wedge \langle x,y_2\rangle\in B_2 ) \}$.
 Now, if $x\in L_1 \bigtriangleup L_2$ then either $(a_1\geq 1 \mbox{ and } a_2 < 0)$ or $(a_1 < 0 \mbox{ and } a_2 \geq 1)$ then $\ell >2^{2P(n,k)-1}$ and 
 if $x\notin L_1 \bigtriangleup L_2$ then either both $a_1$ and $ a_2 $ are greater than equal to 1 or both are less than 1, and in both the cases $\ell\leq 2^{2P(n,k)-1} $. Let $M'$ be a  $k$-restricted Turing machine that on input
 $(x,k)$ guesses two  strings $y_1$ and $y_2$ of length $P'(n,k)$ each, and queries $(x,k,y_i)$ to $B_i'$, $1\le i\le 2$, accepts 
 if and only if exactly one of the oracle answers is YES. It can be seen that $\acc_{M'}(x,k)=\ell(x,k)$. 
We conclude $L_1 \bigtriangleup L_2 \in \ppp$. 
 \end{proof}

\section{ Polynomial Identity Testing}
M\"uller~\cite{Mul08} studied the Arithmetic Circuit Identity Testing (ACIT)  problem with various parameters and obtained  upper bounds  as well as 
hardness results for each of the parameters considered.  However none of the parameters considered in~\cite{Mul08} seem 
adequate for developing a complexity theory for the parameterized probabilistic and counting classes along the lines of classical complexity classes. 

Recall that, in ACIT we are given an arithmetic circuit $C$ as an input and the task is to test if the polynomial computed by $C$ is identically zero. We consider the degree of the polynomial computed by $C$ as a parameter.
\begin{problem}[$\pacit$]
{\em Input}: Arithmetic circuit $C$, $\sdeg(C)\le k$.\\
{\em Parameter}: $k$. \\
{\em Task}: Test if the polynomial computed by $C$ is identically zero. 
\end{problem}
Our main objective now  is to show that $\pacit\in \rfpt$. However, it should be noted that this does not follow directly from the Schwartz-Zippel Lemma,  since it would require $O(n\log k)$ random bits. So the challenge here is to reduce the number of random bits required to $f(k)\log n$. Towards this, we use a mapping defined by Shpilka and Volkovich~\cite{SV09} that reduces the number of variables from $n$ to $2k$. Then we apply Alon's Combinatorial Nullstellensatz~\cite{Alo99} to obtain  what can be treated as a parameterized version of the Schwartz-Zippel lemma.

We begin with a few observations on polynomials of degree at most $k$. Let $S$ be any finite subset of $\mathbb{K}$ that includes $0\in \mathbb{K}$ and let $W^{k}_{n}(S)$ denote the set of all vectors in $S^n$ with at most $k$ non zero entries i.e, the set of all vectors of Hamming weight at most $k$. 
\begin{lemma}
\label{lem:weightk}
Let $f$ be an $n$-variate polynomial of degree at most $k$. Then 
$$f\equiv 0 \iff \forall a\in W^{k}_n(S)~~f(a)=0,$$
where $S\subset \mathbb{K}$ has at least $k+1$ elements. 
\end{lemma}
\begin{proof}
For simplicity, we denote $W^{k}_n(S)$ by $W^{k}_n$. The proof is by induction on $n$. For the base case, suppose $n\le k$.  Since individual degrees of each variable is bounded by $k$, by Proposition~\ref{pro:combnull}, we have $f\equiv 0 \iff f(a)=0~\forall~a\in S^n$, for an $S$ with $|S|\ge k$.

For the  induction step, let $n > k$, and $f(a)=0~\forall~ a\in W^{k}_n$.
For $i\in\{1,\ldots, n\}$, let  $f_i{=} f|_{x_i=0}$, i.e., $f$ substituted with $x_i=0$.
Note that  each of the $f_i$ is a degree $k$ polynomial on at most $n-1$ variables, and $\forall a\in W^{k}_{n-1}~~f_i(a)=0$. By the induction hypothesis, we have $f_i\equiv 0$, and hence $x_i$ divides $ f$. Repeating the argument for all $i\in [1,\ldots, n]$, we have $x_1x_2\cdots x_n$ divides $f$, and hence ${\sf deg}(f)\ge n > k$, a contradiction since  ${\sf deg}(f)=k<n$. Thus we conclude $\forall a\in W^{k}_n ~f(a)=0\implies f\equiv 0$.  The converse direction  is  trivially true. 
\end{proof} 
 
We need a function introduced by Shpilka and Volkovich~\cite{SV09}, that gives a map $G_k: \mathbb{K}[x_1,\ldots, x_n]\to \mathbb{K}[y_1,\ldots, y_{2k}]$ and  serves as a non-identity preserving for a large class of polynomials.  We observe that $G_k$ also functions as a non-identity preserving map for the class of all $n$ variate polynomials of degree at most $k$. We begin with the definition of the generator $G_k$.
\begin{definition}[Shpilka-Volkovich Hitting set generator,\cite{SV09}]
Let $a_1,\ldots, a_n$ be distinct elements in $\mathbb{K}$. Let $G^{i}_k \in \mathbb{K}[y_1,\ldots, y_k,  z_1,\ldots, z_k]$ be the polynomial defined as follows:
\begin{eqnarray*}
G^{i}_k(y_1,\ldots,y_k, z_1,\ldots, z_k)= \sum_{j=1}^k L_i(y_i)z_i, ~\mbox{ where }~ L_i(x)= \frac{\prod_{j\neq i}(x-a_j)}{\prod_{j\neq i}(a_i-a_j)}.
\end{eqnarray*}
The generator $G_k$ is defined  as $G_k \stackrel{\triangle}{=} (G^{1}_k,\ldots, G^{n}_k)$. 
\end{definition}
\begin{lemma}
\label{lem:image}
For any  finite set $S\subset \mathbb{K}$,
then $
W^{k}_n(S) \subseteq \{(G^{1}_k(a), \ldots, G^{n}_k(a))~ |~ a\in (S\cup\{a_1,\ldots, a_n\})^{2k}\}.$
  
\end{lemma}
\begin{proof}
The proof essentially follows the arguments in~\cite{SV09}. We include a sketch here for the sake of completeness.
Note that,
\[L_i(\alpha)= \begin{cases} 0~~\alpha=a_j, \mbox{ if } j\neq i\\ 1 \mbox{ if } \alpha=a_i.\end{cases}
       \] 
       Thus if we set $y_\ell= a_i$, then the image of $G^{i}_k$ contains $z_i$ as a summand. By ensuring that $y_j, i\neq j$ gets some $a_\ell, i\neq \ell$, we get $G^{i}_k=z_i$. In this way we can obtain all vectors of Hamming weight $k$, by setting $y_i$'s and $z_i$'s accordingly. \qed
\end{proof}
Combining Lemma~\ref{lem:image} with Lemma~\ref{lem:weightk} we have:
\begin{lemma}
\label{lem:genk}
Let $f$ be a polynomial of degree at most $k$. Then $ f\equiv 0 \iff f(G_k)\equiv 0$.
\end{lemma}
%

\begin{theorem}
\label{thm:pacit}
$\pacit$  is in $\rfpt$
\end{theorem}
\begin{proof}
By Lemma~\ref{lem:genk} $\pacit$ reduces to testing identity of $2k$-variate polynomials
of degree $O(nk)$ (since the polynomials $L_i$ have degree  $n$). Now applying the Schwartz-Zippel lemma~\cite{Sch80,Zip79}, we obtain a randomized algorithm that uses $O(2k \log( nk))$ random bits and runs in time polynomial in $n$ and $k$.
\end{proof}

\section{Parameterized Permanent vs  Determinant} 

 The determinant ({\sf det}) permanent ({\sf perm}) functions  are   defined  as 
 \begin{eqnarray}
 {\sf det} (A) &=&\sum_{\sigma \in S_n}\prod_{i=1}^{n}{\sf sgn}(\sigma)a_{i,\sigma(i)}
  \label{eqn:det}
   \\ {\sf perm}(A) &=& \sum_{\sigma \in S_n}\prod_{i=1}^{n}a_{i,\sigma(i)},
\label{eq:perm} 
 \end{eqnarray}
 where $A=(a_{i,j})_{1\le i,j\le n}\in {\mathbb N}^{n\times n}$, $S_n$ is the set of all permutations on $n$ symbols and ${\sf sgn}$  is the sign function for permutations. It is known that, given an integer matrix $A$, computing ${\sf A}$ can be done in polynomial time (e.g., Gaussian Elimination method).  In his celebrated paper, 
  Valiant~\cite{Val79} showed that computing ${\sf perm}$ of an integer even for 0 or 1 matrix is complete for $\ShP$. 
Though there are several natural counting problems that characterize $\#W[1]$, it is desirable to have a parameterized variant of permanent so that we get access to the algebraic properties of the permanent function. 

 Naturally, we expect any parameterized variant of permanent to be a 
 function of degree $k$ in $n^2$ variables, where $k$ is the parameter.   
  One way to achieve this would be to restrict the summation given in (\ref{eq:perm})  that move exactly $k$-elements. Formally, a permutation $\sigma\in S_n$ is said to be a {\em $k$-permutation}, if $|\{i~|~ \sigma(i)\neq i\}| =k$.  Let $S_{n,k}$ denote the set of all $k$-permutations on $n$ symbols. 
 
 \begin{definition}  Let $k$ be a parameter. 
 The parameterized determinant ($\pdet$)  and permanent ($\pperm$) functions of a matrix $A\in {\mathbb Z}^{n\times n}$ are defined as follows:
 \begin{eqnarray*}
  \pdet(A,k)&=&\sum_{\sigma\in S_{n,k}}  \prod_{i\neq \sigma(i)} {\sf sgn(\sigma)}a_{i \sigma(i)}
 \\
 \pperm(A,k)&=&\sum_{\sigma\in S_{n,k}}  \prod_{i\neq \sigma(i)} a_{i \sigma(i)},
 \end{eqnarray*}

where $k$ is a parameter.
By abusing the notation, we also let $\pperm$ denote the problem of computing $\pperm$ of an $n\times n$ matrix, where $k$ is the parameter.
\end{definition}
Quite expectedly, $\pdet$ is FPT and  $\pperm$ can be shown to be $\#W[1]$ complete under fpt-reductions. We start with the tractability of $\pdet$.
\begin{theorem}
\label{thm:pdet}
 ${\pdet}$ on integer matrices  is fixed parameter tractable.
\end{theorem}
\begin{proof}

Let $A\in \mathbb{Z}^{n\times n}$, and $k$ be the parameter. Let $A'$ be the matrix obtained from $A$ by replacing the diagonal entries in $A$ by zeroes. Clearly $\pdet(A,k)= \pdet(A',k)$. Let $x$ be a formal variable.  Then ${\sf det}(xA')$ is a univariate polynomial of degree bounded by $n$, and the coefficient of $x^k$ in ${\sf det}(xA')$ is equal to $\pdet(A,k)$.  The value $\pdet(A,k)$ be recovered using the standard interpolation of univariate polynomials.   
\qed   \end{proof}

\begin{theorem}
\label{thm:pperm}
$\pperm$ on matrices in $\mathbb{N}^{n\times n}$ is $\#W[1]$ complete. The hardness holds even in the case of 0-1 matrices. 
\end{theorem}
\begin{proof} It is  known that  counting $k$-matchings in a bipartite graph is complete for  $\#\w{1}$ even in the weighted case~\cite{CM14}.   We prove a parameter preserving equivalence between $\pperm$ and the problem of counting $k$-matchings in a bipartite graph  which  completes the proof.
 For the upper bound, we give a reduction from $\pperm$ to counting $k$-matchings in a bipartite graph. For a given matrix $A\in \mathbb{N}^{n\times  n}$  define the matrix $A'$ by setting the diagonal entries of $A$ to zero, i.e., $A'[i,j]= A[i,j]$ if $i\neq j$ and $A'[i,i]=0,$ $ 1\le i\le n$. Note that $\pperm(A)= \pperm(A')$.  Every $k$-permutation of $\{1,\ldots, n\}$ corresponds to a matching of size $k$ in the bipartite graph $G'$ with $A'$ is the bipartite adjacency matrix. Thus $\pperm(A')=$the sum  of weights of $k$-matchings in  $G'$.

For the hardness we give a reduction in the reverse direction, i.e.,  a parameter preserving reduction from counting the number of $k$- matchings in a Bipartite graph $G=(U,V,E)$ to computing $\pperm$ of an integer matrix. 

 Let $G=(U, V, E)$ be a  bipartite graph. Without loss of generality, assume that $U=V=\{1,\ldots, n\}$.   Construct a new bipartite graph   $G'= (U',V', E')$ with $U'=V'=\{1,\ldots, 2n\}$.  For every edge of the form $(i,j)\in E, i\neq j$, $G'$ contains the edge $(i,j)\in E'$. For edges of the form $(i,i)\in E$, $G'$ contains the edge $(i,n+i)\in E'$.  Note that the vertices $n+1,\ldots, 2n$ in $U'$  are isolated vertices.

 Note  that the  number of matchings in $G$ and  those in $G'$  of a given size $k$ are equal.
Let $A'$ be the bipartite adjacency matrix of $G'$.
 Every $k$-permutation of $[2n]$  that contributes  a non-zero value to $\pperm(A')$  corresponds to a matching of size $k$ in $G'$. Moreover, none of the $k$-matchings in $G'$ will have an edge of the form $(i,i), i\in [2n]$. Thus, $\pperm(A',k)= \#$matchings of size $k$ in $G'$.  This completes the proof.
\qed   \end{proof}
\section*{Conclusions}
We have studied  parameterized variants of probabilistic computation.
We hope that our definition of $\ppp$ leads to further developments in the structural aspects of probabilistic and counting complexities in the parameterized world. Further, $\ppp$ might be useful in  defining a   parameterized variant of the Counting Hierarchy (CH) which could in turn have implications to parameterized complexity of numerical and algebraic computation~\cite{ABKM09}.  Though definition of a parameterized CH  based on $\ppp$ is straightforward, the  usefulness of such a definition would rely on $\ppp$ being closed under intersection, which is not known currently. 

Further, we believe any  fixed parameter tractable randomized algorithm should naturally place the problem in $\Wp$. One way to achieve this is to   obtain  randomized FPT algorithms  that use at most $O(f(k)\log n)$ random bits. As a first step towards this direction, we introduce  a natural parameterization  to the polynomial identity testing for which we obtain such an algorithm.   We hope our observations will lead to further development of randomness efficient parameterized algorithms.

\section*{Acknowledgements}
We thank anonymous reviewers for their comments on an earlier version of this paper which helped in improving the presentation of the article.
\bibliographystyle{abbrv}
\bibliography{reports}
\newpage
\appendix
\section{Proof of Theorem~\ref{thm:pp-comp}}
Let L $\in \ppp$ then there exist a $k$-restricted   Turing machine $M$ running  using at most $P(n,k)$ random bits such that
\begin{eqnarray*}
(x,k) \in L &\Rightarrow& Pr_{y\in \{0,1\}^{f(k)\log n}}[M\mbox{ accepts }(x,k)]\geq \frac{1}{2}+\frac{1}{2^{P(n,k)}}; \mbox{ and }\\
(x,k) \notin L &\Rightarrow& Pr_{y\in \{0,1\}^{f(k)\log n}}[M\mbox{ accepts }(x,k)]\leq \frac{1}{2}
\end{eqnarray*}
  Let  $M'$ be the machine that on input $(x,k)$ simulates $M$, rejects $(x,k)$ if $M$ does so, and whenever $M$ accepts,  chooses a random string of length $P(n,k)+1$ and  rejects only if the random string is $1^{P(n,k)+1}$ and accepts otherwise. For any $(x,k)\in \Sigma^*\times k$, we have:   
\begin{eqnarray*}
(x,k) \in L &\Rightarrow& Pr_{y\in \{0,1\}^{P(n,k)}}[M'\mbox{ accepts }(x,k)]\geq (\frac{1}{2}+\frac{1}{2^{P(n,k)}})(1-\frac{1}{2^{p(n,k)+1}})\\ & & ~~~~~~~~~~~~~~~~~~~~~~~~~~~~~~~~~~~~~~~~~~~~~~~~~~~~~~~~
>\frac{1}{2}; \mbox{ and } \\
 (x,k) \notin L &\Rightarrow& Pr_{y\in \{0,1\}^{f(k)\log n}}[M' \mbox{ accepts } (x,k)]\leq \frac{1}{2}(1-\frac{1}{2^{p(n,k)+1}})<\frac{1}{2} 
\end{eqnarray*}
Now, let $M^c$ be the machine that flips the answers of $M'$, i.e., $M^c$ accepts whenever $M'$ rejects and vice versa. We have :
\begin{eqnarray}
 (x,k) \notin \overline{L}& \Rightarrow& Pr[ M^c\mbox { accepts } (x,y,k)]=Pr[ M^c \mbox{ rejects } (x,y,k)] < \frac{1}{2} \\ 
(x,k) \in \overline{L} &\Rightarrow& Pr[M^c \mbox{ accepts  }(x,y,k)]=Pr[M \mbox{ rejects }(x,y,k)]  >  \frac{1}{2} 
\end{eqnarray}
  Note that $M'$ is $k$ restricted.   This completes the proof.\qed
\end{document}